\documentclass[leqno,centertags,cmex10]{amsart}

\usepackage{amsfonts,amssymb,curves}

\newtheorem{theorem}{Theorem}
\newtheorem{proposition}[theorem]{Proposition}
\newtheorem{lemma}[theorem]{Lemma}
\newtheorem{example}[theorem]{Example}
\newtheorem{remark}[theorem]{Remark}
\newtheorem{corollary}[theorem]{Corollary}

\newcommand{\cal}{\mathcal}

\begin{document}

\title{Generalizations of Wei's Duality Theorem}

\author[Britz]{Thomas Britz}
\address{School of Mathematics and Statistics,
University of New South Wales, Sydney, NSW 2052, Australia}
\email{britz@unsw.edu.au}
\thanks{T.~Britz was supported by an ARC Discovery Grant.}

\author[Heiseldal]{B{\aa}rd Heiseldal}
\address{AXTech AS, Verftsgt.~10, PO Box 2008, 6402 Molde, Norway}
\email{Bard.Heiseldal@axtech.no}

\author[Johnsen]{Trygve Johnsen}
\address{Department of Mathematics and Statistics, University of
Troms{\o}, N-9037 Troms{\o}, Norway}
\email{trygve.johnsen@uit.no}

\author[Mayhew]{Dillon Mayhew}
\address{School of Mathematics, Statistics and Operations Research,
Victoria University, PO Box 600, Wellington 6140, New Zealand}
\email{dillon.mayhew@msor.vuw.ac.nz}

\author[Shiromoto]{Keisuke Shiromoto}
\address{Department of Mathematics and Engineering,
Kumamoto University, 2-39-1, Kurokami, Kumamoto 860-8555, Japan}
\email{keisuke@kumamoto-u.ac.jp}

\date{\today}

\begin{abstract}
Wei's celebrated Duality Theorem is generalized in several ways,
expressed as duality theorems for linear codes over division rings
and, more generally, duality theorems for matroids.
These results are further generalized,
resulting in two Wei-type duality theorems for new combinatorial structures
that are introduced and named {\em demi-matroids}.
These generalize matroids and are the appropriate combinatorial objects
for describing the duality in Wei's Duality Theorem.
A new proof of the Duality Theorem is thereby given
that explains the theorem in combinatorial terms.
Special cases of the general duality theorems are also given,
including duality theorems for cycles and bonds in graphs
and for transversals.
\end{abstract}

\maketitle


\section{Introduction}
In 1991, Victor K.~Wei presented his remarkably simple and elegant Duality Theorem~\cite[Theorem~3]{wei91}
concerning the generalized Hamming weights (higher weights) of a code over a finite field and of its dual code.
This result has attracted much favorable attention among coding theoreticians
and has created an active sub-field of research,
while also re-invigorating the studies of higher weight enumerators (support weight enumerators)
for linear codes over finite fields (see~\cite{ChKl03,DoGu01,HeKlMy77,HoSh99,kloeve92,MiCoCo03,milenkovic05,shiromoto03,TsVl95} for instance).

Wei's Duality Theorem has been re-proved and generalised by several authors,
often with respect to codes over certain finite rings
(see \cite{Ashikhimin98,heiseldal08,HoSh01,schaathunthesis} for a few generalizations).
In the present paper,
the Duality Theorem is generalized with respect to matroids, and more generally, to demi-matroids.
Matroids capture the combinatorial essence of linear independence,
and recent research has demonstrated how matroid theory may be applied with good effect to coding-theoretical problems
(see~\cite{BrBrShSo07,britz07a}).
More general than matroids are the demi-matroids,
which turn out to be the appropriate combinatorial structures for describing the duality in Wei's Duality Theorem.
Section~\ref{sec:defnnot} introduces demi-matroids and describes some of their parameters and properties.
In addition to the matroid dual, each demi-matroid also has another type of dual, its {\em supplement}.
The dual and the supplement each give rise to a generalization of Wei's Duality Theorem
and thereby explain in combinatorial terms why Wei's Duality Theorem holds.

Special cases of these generalizations are presented in Section~\ref{sec:specialcases}.
The first special case is Wei's Duality Theorem generalized with respect to matroids.
In turn, this result has as special cases two duality results,
namely for graphs (first proved in~\cite{britz07a}) and for transversals.
In each of these two cases,
the matroid-analogues of the generalized Hamming weights have natural graph- or transversal interpretations.
A final special case is a duality theorem
concerning the generalized Hamming weights of a code over a division ring and of its dual code;
this generalizes Wei's Duality Theorem and provides a new proof thereof.

This paper is largely self-explanatory but does at times refer to some elementary matroid theory.
For information on matroids, see the excellent expositions~\cite{oxley06,welsh76}.

\section{Wei-type duality theorems for demi-matroids}
\label{sec:defnnot}

A {\em demi-matroid} is a triple $(E,s,t)$ consisting of
a set $E$ and two functions $s,t:2^E\longmapsto\mathbb{N}_0$
satisfying the following two conditions for all subsets $X\subseteq Y\subseteq E$:
\begin{list}{}{\leftmargin=1em\topsep=1.5mm\itemsep=1mm}
\item[(R)] $0\leq s(X)\leq s(Y)\leq |Y|$ and $0\leq t(X)\leq t(Y)\leq |Y|$;
\item[(D)] $|E-X|-s(E-X)=t(E)-t(X)$.
\end{list}
Note that $s(\emptyset)=t(\emptyset)=0$ by (R).
It follows that (D) is equivalent to the following condition:
\begin{list}{}{\leftmargin=1em\topsep=1.5mm\itemsep=1mm}
\item[(D')] $|E-X|-t(E-X)=s(E)-s(X)$.
\end{list}

Note that for any matroid $M$ on $E$ with rank function $\rho$,
the triple $(E,\rho,\rho^*)$ is a demi-matroid.
Conversely, if $(E,s,t)$ is a demi-matroid,
then $s$ is the rank function of a matroid $M$ on $E$ if and only if $t$ is the rank function of $M^*$.
The following example shows that demi-matroids properly generalize matroids.
\begin{example}
\label{exa:not_matroid}
{\rm
Suppose that $E=\{a,b\}$ and define $s(X):=0$ for $X=\emptyset,\{a\},\{b\}$, and $s(E):=1$.
The triple $(E,s,s)$ is a demi-matroid but $s$ is not the rank function of any (poly)matroid on $E$.}
\end{example}

Let $E$ be a set of $n$ elements, and let $D=(E,s,t)$ be a demi-matroid.
By~(D),\vspace{-.5mm}
\[
s(E)+t(E)=n\,.\vspace{-.5mm}
\]
Set $k:=s(E)$.
\begin{lemma}
\label{lem:difference}
$s(X-x)\geq s(X)-1$ and $t(X-x)\geq t(X)-1$
for all $X\subseteq E$ and $x\in E$.
\end{lemma}
\begin{proof}
By (R) and (D),
\begin{align*}
t(X-x)
&=    t(E) - |E-(X-x)| + s(E-(X-x))\\
&\geq t(E) - |E-X| - 1 + s(E-X)\\
&=    t(X) - 1\,.
\end{align*}
Similarly,
$s(X-x)\geq s(X)-1$.
\end{proof}

Define for all $i=0,\ldots,k$ and $j=0,\ldots,n-k$,\label{equ:stst_def}
\begin{align*}
\sigma_i &:= \min \{\, |X|  \::\:  X\subseteq E,\: s(X) \geq i\}\,;\\
\tau_j   &:= \min \{\, |X|  \::\:  X\subseteq E,\: t(X) \geq j\}\,;\\[1mm]
  s_i    &:= \max \{\, |X|  \::\:  X\subseteq E,\: s(X) \leq i\}\,;\\
  t_j    &:= \max \{\, |X|  \::\:  X\subseteq E,\: t(X) \leq j\}\,.
\end{align*}
By (R) and Lemma~\ref{lem:difference},
all of the numbers $\sigma_i$, $\tau_i$, $s_i$, and $t_j$ are well-defined
and may be given the following equivalent characterizations:
\begin{lemma}
\label{lem:altdef}
For all $i=0,\ldots,k$ and $j=0,\ldots,n-k$,
\begin{align*}
\sigma_i &= \min \{\, |X|  \::\:  X\subseteq E,\: s(X) = i\}\,;\\
\tau_j   &= \min \{\, |X|  \::\:  X\subseteq E,\: t(X) = j\}\,;\\
  s_i    &= \max \{\, |X|  \::\:  X\subseteq E,\: s(X) = i\}\,;\\
  t_j    &= \max \{\, |X|  \::\:  X\subseteq E,\: t(X) = j\}\,.
\end{align*}
\end{lemma}
\begin{remark}
\label{rem:matroid-trivial}
{\rm If $M$ is a matroid on $E$ with rank function~$\rho$,
then the coefficients $\sigma_i,\tau_j$ for demi-matroid $D:=(E,\rho,\rho^*)$ are trivial:
$\sigma_i=i$ and $\tau_j=j$ for all $i,j$.}
\end{remark}
\begin{lemma}
\label{lem:monotonicity}
The following inequalities hold:
\begin{align*}
0&=\sigma_0<\sigma_1<\sigma_2<\cdots<\sigma_k\leq n\,;\\
0&=\tau_0<\tau_1<\tau_2<\cdots<\tau_{n-k}\leq n\,.
\end{align*}
\end{lemma}
\begin{proof}
For each $i=1,\ldots,k$,
let $X \subseteq E$ be a subset such that $|X|=\sigma_i$ and $s(X) \geq i$.
By Lemma~\ref{lem:difference},
$s(X-x)\geq i-1$ for any $x\in X$,
so $\sigma_{i-1} \leq |X-x| <\sigma_i$.

Similarly, $\tau_{j-1} <\tau_j$ for each $j=1,\ldots,n-k$.
\end{proof}

\begin{lemma}
\label{lem:monotonicity-min}
The following inequalities hold:
\begin{align*}
0&\leq s_0<s_1<s_2<\cdots<s_k=n\,;\\
0&\leq t_0<t_1<t_2<\cdots<t_{n-k}=n\,.
\end{align*}
\end{lemma}
\begin{proof}
Similar to the proof of Lemma~\ref{lem:monotonicity}.
\end{proof}

The four above monotonicities each induce a generalized Singleton-type bound for demi-matroids:
\begin{corollary}
For all $i=0,\ldots,k$ and $j=0,\ldots,n-k$,
\begin{align*}
\sigma_i &\leq n-k+i\,;\\
s_i      &\leq n-k+i\,;\\
\tau_j   &\leq k+j  \,;\\
t_j      &\leq k+j  \,.
\end{align*}
\end{corollary}

The {\em dual demi-matroid} of
a demi-matroid $D:=(E,s,t)$ is the triple $D^*:=(E,t,s)$.
The operation $D \longmapsto D^*$ is clearly an involution,
i.e,
\[
D=(D^*)^*\,.
\]
A second fundamental involution on demi-matroids is now presented.
For any real function $f:2^E\longmapsto\mathbb{R}$,
let $\overline{f}$ denote the function given by
\[
 \overline{f}(X) := f(E)-f(E-X)\,.
\]
Since
\[
 \overline{\overline{f}}(X)
=\overline{f}(E)-\overline{f}(E-X)
=f(X)-f(\emptyset)\,,
\]
it follows that if $f(\emptyset)=0$, then the operation $f\longmapsto\overline{f}$ is an involution,
i.e, $f=\overline{\overline{f}}$.
\begin{theorem}
\label{thm:involution}
The triple $\overline{D}:=(E,\overline{s},\overline{t})$ is a demi-matroid;
furthermore, $D=\overline{\overline{D}}$ and $\overline{D^*}=\overline{D}^*$.
\end{theorem}
\begin{proof}
To show that $\overline{D}$ is a demi-matroid,
first note
that $\overline{s}(\emptyset)=\overline{t}(\emptyset)=0$ and
that $\overline{s}(E)=s(E)$ and $\overline{t}(E)=t(E)$.
Consider subsets $X\subseteq Y\subseteq E$.
By (R) and (D),
\[
0\leq s(E)-s(E-X)\leq s(E)-s(E-Y)=|Y|-t(Y)\leq|Y|,
\]
so $0\leq\overline{s}(X)\leq\overline{s}(Y)\leq|Y|$.
Similarly, it is easy to show that
$0\leq\overline{t}(X)\leq\overline{t}(Y)\leq|Y|$,
so $\overline{D}$ satisfies $(R)$.
By (D'),
\begin{align*}
  |E-X|-\overline{s}(E-X)
&= |E-X|-(s(E)-s(X))\\
&= t(E-X)\\
&= t(E)-\overline{t}(X)\\
&= \overline{t}(E)-\overline{t}(X)\,,
\end{align*}
so $\overline{D}$ satisfies (D).
Hence, $\overline{D}$ is a demi-matroid.

Finally,
note that
$D = (E,s,t)
   = (E,\overline{\overline{s}},\overline{\overline{t}})
   = \overline{\overline{D}}$
and that
$\overline{D^*}
   = \overline{(E,t,s)}
   = (E,\overline{t},\overline{s})
   = (E,\overline{s},\overline{t})^*
   = \overline{D}^*$.
\end{proof}

The demi-matroid $\overline{D}$ is called the {\em supplement} of~$D$.
\begin{example}
\label{exa:tri}
{\rm The supplement operation does not generally apply to matroids.
For instance, consider the matroid $M:=(E,\rho)$
where $E=\{a,b,c\}$ and $\rho(X)=0$ for $X=\emptyset,\{a\}$ and $\rho(X)=1$ for all other subsets $X\subseteq E$.
Then $D:=(E,\rho,\rho^*)$ is a demi-matroid,
so $\overline{D}=(E,\overline{r},\overline{r^*})$ is also a demi-matroid.
However, $(E,\overline{\rho})$ is not a matroid,
since it would have rank 1 but only contain loops.}
\end{example}

Define for all $i=0,\ldots,k$ and $j=0,\ldots,n-k$,
\begin{align*}
  \overline{\sigma}_i &:= \min \{\, |X|  \::\:  X\subseteq E,\: \overline{s}(X) = i\}\,;\\
  \overline{\tau}_j   &:= \min \{\, |X|  \::\:  X\subseteq E,\: \overline{t}(X) = j\}\,;\\[2mm]
  \overline{s}_i      &:= \max \{\, |X|  \::\:  X\subseteq E,\: \overline{s}(X) = i\}\,;\\
  \overline{t}_j      &:= \max \{\, |X|  \::\:  X\subseteq E,\: \overline{t}(X) = j\}\,.
\end{align*}
\begin{lemma}
\label{lem:equ}
For each $i=0,\ldots,k$ and $j=0,\ldots,n-k$,
\begin{align*}
  s_i      &= n-\overline{\sigma}_{k-i}\,;\\
  \sigma_i &= n-\overline{s}_{k-i}\,;\\
  t_j      &= n-\overline{\tau}_{n-k-j}\,;\\
  \tau_j   &= n-\overline{t}_{n-k-j}\,.
\end{align*}
\end{lemma}
\begin{proof}
By Lemma~\ref{lem:altdef},
\begin{align*}
  s_i &= \max\{|X|   \::\: X\subseteq E,\: s(X) = i\}\\
      &= \max\{|E-X| \::\: X\subseteq E,\: s(E-X) = i\}\\
      &= n-\min\{|X| \::\: X\subseteq E,\: \overline{s}(X) = k-i\}\\
      &= n-\overline{\sigma}_{k-i}\,.
\end{align*}
The remaining identities are proved similarly.
\end{proof}

For each demi-matroid $D$,
set
\begin{align*}
S_D  &:= \{n-s_{k-1},\ldots,n-s_1,n-s_0\}\,;\\
T_D  &:= \{t_0+1,t_1+1,\ldots,t_{n-k-1}+1\}\,;\\[2mm]
U_D  &:= \{\sigma_1,\sigma_2,\ldots, \sigma_k\}\,;\\
V_D  &:= \{n+1-\tau_{n-k},\ldots,n+1-\tau_2, n+1-\tau_1 \}\,.
\end{align*}
Lemma~\ref{lem:equ} implies the following identities:
\begin{lemma}
\label{lem:setequ}
$S_D = U_{\overline{D}}$ and $T_D = V_{\overline{D}}$.
\end{lemma}
The main results of this paper are the following fundamental duality theorems
for demi-matroids that each generalize Wei's Duality Theorem~\cite{wei91}.
\begin{theorem}
\label{thm:demi-matroidWei-equ}
$S_D\cup T_D=\{1,\ldots,n\}$ and $S_D\cap T_D=\emptyset$.
\end{theorem}
\begin{theorem}
\label{thm:demi-matroidWei}
$U_D\cup V_D=\{1,\ldots,n\}$ and $U_D\cap V_D=\emptyset$.
\end{theorem}
\begin{proof}[Proof of Theorems~\ref{thm:demi-matroidWei-equ} and~\ref{thm:demi-matroidWei}]
Assume that there are integers $i,j$ such that $\sigma_i=n+1-\tau_j$.
Let $X\subseteq E$ be a subset satisfying $|X|=\tau_j$ and $t(X) \geq j$.
Then $|E-X|=\sigma_i-1$, so $s(E-X)\leq i-1$ from Lemma \ref{lem:monotonicity}.
By (D),
\[
n-\tau_j-(n-k)+j =-\tau_j+k+j \leq i-1 \,.
\]
Similarly,
\[
n-\sigma_i-k+i \leq j-1 \,.
\]
Hence,
$-1=n-\sigma_i-\tau_j\leq -2$,
a contradiction.
This proves Theorem~\ref{thm:demi-matroidWei}.

To prove Theorem~\ref{thm:demi-matroidWei-equ},
apply Theorem~\ref{thm:demi-matroidWei} to $\overline{D}$,
and use Lemma~\ref{lem:setequ}.
\end{proof}

\begin{example}
\label{exa:dualities}
{\rm
For the demi-matroid $D:=(E,s,t)$ with $E=\{a,b,c\}$, $s(E)=1$, and $s(X)=0$ for~$X\subsetneq E$,
\[
\begin{array}{rlrl}
S_D &= \{1\}\quad & U_D &= \{3\}\,;\\
T_D &= \{2,3\}    & V_D &= \{1,2\}\,.
\end{array}
\]
Thus, $S_D\cup T_D = \{1,2,3\}$
  and $S_D\cap T_D = \emptyset$,
as asserted by Theorem~\ref{thm:demi-matroidWei-equ}.
Similarly,
      $U_D\cup V_D = \{1,2,3\}$
  and $U_D\cap V_D = \emptyset$,
as asserted by Theorem~\ref{thm:demi-matroidWei}.}
\end{example}

\section{Duality theorems for matroids, graphs, transversals, and linear codes}
\label{sec:specialcases}

Let $M=(E,\rho)$ be a matroid of rank~$k:=\rho(M)$ on the set $E$.
For all $i=0,\ldots,k$ and $j=0,\ldots,n-k$,
define
\begin{align*}
f_i  &:= \max \{\, |F|   \::\: F\subseteq E,\: \rho(F)=i\}\,;\\
f^*_j&:= \max \{\, |F|   \::\: F\subseteq E,\: \rho^*(F)=j\}\,.
\end{align*}
Set\vspace{-1mm}
\begin{align*}
S_M  &:= \{n-f_{k-1},\ldots,n-f_1,n-f_0\}\,;\\
T_M  &:= \{f^*_0+1,f^*_1+1,\ldots,f^*_{n-k-1}+1\}\,.
\end{align*}
The following duality result for matroids follows immediately from Theorem~\ref{thm:demi-matroidWei}.
\begin{theorem}
\label{thm:matroidWei}
$S_M\cup T_M = \{1,\ldots,n\}$ and $S_M\cap T_M = \emptyset$.
\end{theorem}

\begin{example}
\label{exa:V8}
{\rm The non-representable {\em V{\'a}mos matroid} $V_8$ on $E:=\{1,\ldots,8\}$
has as its bases $\mathcal{B}(M)$ all 4-subsets of~$E$
except for the following:
\[
\{1,2,5,6\},\{1,3,5,7\},\{1,4,5,8\},\{2,3,6,7\},\{2,4,6,8\}\,.
\]
The matroid $M:=V_8$ is simple, self-dual, and paving,
so $$(f_0,f_1,f_2,f_3)=(f^*_0,f^*_1,f^*_2,f^*_3)=(0,1,2,4).$$
It follows that
$S_M = \{4,6,7,8\}$ and $T_M = \{1,2,3,5\}$.
Thus,
$S_M\cap T_M = \emptyset$
and
$S_M\cup T_M = \{1,\ldots,8\}$,
as~asserted by Theorem~\ref{thm:matroidWei}.}
\end{example}

\begin{remark}
\label{rem:nottwodualities}{\rm
The demi-matroid $D:=(E,\rho,\rho^*)$ satisfies
$\sigma_i=i$ and $\tau_j=j$ for all
$i=0,\ldots,k$ and $j=0,\ldots,n-k$.
Thus, $U_D=\{1,2,\ldots,k\}$ and $V_D=\{k+1,\ldots,n-1,n\}$.
Hence, there is no interesting matroid analogue of Theorem~\ref{thm:demi-matroidWei-equ}.}
\end{remark}
The coefficients $f_i$ and $f^*_j$ can fairly easily be re-expressed in terms of cocircuits and circuits (see~\cite[p.~306]{white86}):
\begin{align}
\notag
n-f_{i-1}   = \min\bigl\{|X|\,:\; &X=\bigcup_{j=1}^iB_j\textrm{ where, for all $j\leq i$,}\\\label{equ:1}\tag{F}
 &B_j\in\mathcal{C}^*(M)\,,\; B_j\nsubseteq\bigcup_{k\neq j}B_k\bigr\}\,;\!\!\\\notag
n-f^*_{j-1} = \min\bigl\{|X|\::\: &X=\bigcup_{i=1}^jC_i\textrm{ where, for all $i\leq j$,}\\\label{equ:2}\tag{F${}^*$}
 &C_i\in\mathcal{C}(M)  \,,\; C_i\nsubseteq\bigcup_{k\neq i}C_k\bigr\}\,.
\end{align}
These identities will be used in the subsections below.

\subsection{Perfect matroid designs}
\label{sec:pmd}

A {\em perfect matroid design} is a matroid~$M$
in which the cardinality of each closed set is determined uniquely by its rank
(see~\cite[Chapter~12]{welsh76}).
If the rank of a closed set $F$ of $M$ is~$i$, then $|F|=f_i$.
Theorem~\ref{thm:matroidWei} immediately implies the following result.
\begin{corollary}
\label{cor:designmatroids}
The cardinalities of the closed sets of a perfect matroid design~$M$
are uniquely determined by the closed set cardinalities of~$M^*$.
\end{corollary}

\subsection{Graphs}
\label{sec:graphs}

Let $G$ be a (multi)graph on~$n$ edges whose spanning forests each contains $k$ edges.
Recall that a {\em bond} of $G$ is a minimal cut-set of edges of $G$.
For each $i=1,\dots,k$ and $j=1,\ldots,n-k$,
define
\begin{align*}
b_i := &\textrm{ minimal number of edges in a union of $i$ bonds,}\\
       &\textrm{ none contained in the union of the others}\,;\\[1mm]
c_j := &\textrm{ minimal number of edges in a union of $j$ cycles,}\\
       &\textrm{ none contained in the union of the others.}
\end{align*}
Consider the cycle matroid $M:=M(G)$ and its coefficients $f_i$ and $f^*_j$.
Equations (\ref{equ:1}) and (\ref{equ:2}) immediately imply the following result:
\begin{proposition}
\label{prop:n-f}
$b_i:=n-f_{i-1}$ and $c_j:=n-f^*_{j-1}$.
\end{proposition}
\begin{corollary}
\label{cor:forests}
The maximal number of edges in any subgraph of $G$ whose spanning forests each contain $i-1$ edges
is $n-b_i$.
Similarly,
$n-c_j$ is the maximal size of an edge set $E'\subseteq E(G)$ for which
$G\backslash E''$ does not span $G$ for any $j$-element subset $E''\subseteq E'$.
\end{corollary}

Set
\begin{align*}
U_G &:= \{b_1,\ldots,b_k\}\,;\\
V_G &:= \{n+1-c_{n-k},\ldots,n+1-c_1\}\,.
\end{align*}
The next result was proved in \cite{britz07a}
and follows immediately from Theorem~\ref{thm:matroidWei} and Proposition~\ref{prop:n-f}.
\begin{theorem}
\label{thm:weigraph}
$U_G\cup V_G=\{1,\ldots,n\}$ and $U_G\cap V_G=\emptyset$.
\end{theorem}

\begin{example}
\label{exa:cyclesandbonds}
{\rm The graph $G$ below has $n=5$ edges and each of its spanning forests has $k=3$ edges:
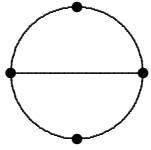
\setlength{\unitlength}{5pt}
\begin{center}
\begin{picture}(10,13)(0,-1)
  \put(5,5){\arc(5,0){360}}
  \curve(0,5,10,5)
  \multiput(5,0)(0,10){2}{\circle*{.9}}
  \multiput(0,5)(10,0){2}{\circle*{.9}}
\end{picture}
\end{center}
For this graph, $(b_1,b_2,b_3)=(2,4,5)$ and $(c_1,c_2)=(3,5)$.
Set\vspace{-1mm}
\begin{align*}
U_G &:= \{b_1,b_2,b_3\}=\{2,4,5\}\,;\\
V_G &:= \{n+1-c_2,n+1-c_1\}=\{1,3\}\,.
\end{align*}
Then $U_G\cup V_G=\{1,2,3,4,5\}$ and $U_G\cap V_G=\emptyset$,
as asserted by Theorem~\ref{thm:weigraph}.}
\end{example}

\begin{example}
\label{exa:complete}
{\rm For the complete graph $G:=K_m$ on $m$ vertices,
the number of edges is $n=\binom{m}{2}$,
and each spanning tree contains $k=m-1$ edges.
In~\cite{heiseldal08}, it was shown that
\begin{align*}
\{b_1,\ldots,b_k\}     &= \textstyle\{n-\binom{i}{2}\::\: i=1,\ldots,k\}\,;\\
\{c_1,\ldots,c_{n-k}\} &= \textstyle\{1,\ldots,n\}\backslash\bigl\{\binom{i}{2}+1\::\: i=1,\ldots,k\bigr\}\,,
\end{align*}
respectively,
so $U_G\cup V_G=\{1,\ldots,n\}$ and $U_G\cap V_G=\emptyset$,
as asserted by Theorem~\ref{thm:weigraph}.}
\end{example}

\begin{example}
\label{exa:bipartite}
{\rm For the complete bipartite graph $G:=K_{l,m}$ with $l\geq m$,
the number of edges is $n=lm$,
and each spanning tree contains $k=l+m-1$ edges.
In~\cite{heiseldal08}, it was shown that
\begin{align*}
\{b_1,\ldots,b_k\}     ={} &\textstyle \{m,2m,\ldots,(l-m)m\}\,\cup\\
    &\textstyle \bigl\{n-\bigl\lfloor\frac{i^2}{4}\bigr\rfloor\::\:i=1,\ldots,2m-1\bigr\}\,;\\
\{c_1,\ldots,c_{n-k}\} ={} &\textstyle \{1,\ldots,n\}\backslash\\
    &\bigl(\{n+1-im\::\: i=1,\ldots,l-m\}\,\cup\\
    &\textstyle\:\:\bigl\{\bigl\lfloor\frac{i^2}{4}\bigr\rfloor+1\::\:i=1,\ldots,2m-1\bigr\}\bigr)\,,
\end{align*}
respectively,
so $U_G\cup V_G=\{1,\ldots,n\}$ and $U_G\cap V_G=\emptyset$,
as asserted by Theorem~\ref{thm:weigraph}.}
\end{example}

\subsection{Transversals}
\label{sec:matchings}

Let ${\cal A}:=\{A_1,\ldots,A_m\}$ be a multiset of subsets $A_j\subseteq E$.
A {\em transversal} of~${\cal A}$ is a set $T\subseteq E$ of size $|T|=|{\cal A}|$
for which the elements of $T$ may be labeled $e_1,\ldots,e_m$
so that $e_j\in A_j$ for each $j=1,\ldots,m$.
A {\em partial transversal} of ${\cal A}$ is a transversal of a sub-multiset of~${\cal A}$.
The partial transversals of $\cal A$ form the independent sets of the {\em transversal matroid} of~${\cal A}$,
denoted by $M[{\cal A}]$ (cf.~\cite[Section 1.6]{oxley06}).
A set $X\subseteq E$ is a {\em plug} for $\cal A$ if
$X-e$ is a partial transversal of $\cal A$ for each $e\in X$
but $X$ itself is not.
Let $k$ denote the maximal size of a partial transversal of $\cal A$, that is, the rank of $M[{\cal A}]$.
For each $i=0,\dots,k-1$ and $j=1,\ldots,n-k$,
define
\begin{align*}
m_i &:= \max \{\,|X| \::\: X\subseteq E,\:\textrm{$X$ contains a partial transversal}\\
    &     \qquad\qquad\qquad\quad \textrm{of $\cal A$ of size $i$ but none of size $i+1$}\}\,;\\
p_j &:= \textrm{minimal size of a union of $j$ plugs for $\cal A$,}\\
    &   \quad\:\:\: \textrm{none contained in the union of the others.}
\end{align*}
Set\vspace{-1mm}
\begin{align*}
U_{\cal A} &:= \{m_0+1,\ldots,m_{k-1}+1\}\,;\\
V_{\cal A} &:= \{p_1,\ldots,p_{n-k}\}\,.
\end{align*}
\begin{theorem}
\label{thm:weimatching}
$U_{\cal A}\cup V_{\cal A}=\{1,\ldots,n\}$ and $U_{\cal A}\cap V_{\cal A}=\emptyset$.
\end{theorem}
\begin{proof}
For $M:=M[{\cal A}]$,
$m_i=f_i$ and $p_j=n-f^*_{j-1}$, by~(\ref{equ:2}).
Apply Theorem~\ref{thm:matroidWei}.
\end{proof}
\begin{example}
\label{exa:trasnversal}
{\rm Let $E:=\{a,b,c,d,e\}$ and
\[
{\cal A}:=\bigl\{\{a,b\},\{a,c\},\{d\},\{d\}\bigr\}\,.
\]
Then
$U_{\cal A}=\{2,4,5\}$ and $V_{\cal A}=\{1,3\}$,
so $U_{\cal A}\cap V_{\cal A}=\emptyset$ and $U_{\cal A}\cup V_{\cal A}=\{1,2,3,4,5\}$,
as claimed by Theorem~\ref{thm:weimatching}.}
\end{example}

\subsection{Codes over division rings}
\label{sec:divisionrings}
Let $R$ denote a division ring (perhaps a field) and set $E:=\{1,2,\ldots,n\}$.
The {\em support} of each codeword $\mathbf{x}:=(x_1,x_2,\ldots,x_n) \in R^n$ is the set
\[
\textrm{supp}(\mathbf{x}) := \{\,i \:\::\:\: x_i \neq 0\}\,.
\]
Similarly, the {\em support} and {\em weight} of each subset $D\subseteq R^n$ are defined as follows:
\begin{align*}
\textrm{Supp}(D)  &:= \bigcup_{\mbox{\small$\mathbf{x}\in D$}} \textrm{supp}(\mathbf{x})\,;\\
\textrm{wt}(D)    &:= |\,\textrm{Supp}(D) |\,.
\end{align*}
Let $C$ be a right linear $[n,k]$ code over~$R$ with coordinates~$E$.
The {\em dual code} $C^\perp$ is given as follows:
\[
C^\perp=\{\mathbf{y}\in R^n \::\: \mathbf{x} \cdot \mathbf{y} =0\,, \: {}^\forall \mathbf{x} \in C\}.
\]
For each integer $i=1,\ldots,k$ ($j=1,\ldots,n-k$),
define the $i$th ($j$th) {\em generalized Hamming weight} of $C$ ($C^\perp$) as follows:
\begin{align*}
d_i       &:= \min\bigl\{\textrm{wt}(D)\,:\, \textrm{$D$ a right linear $[n,i]$ subcode of $C$}\}\,;\\
d^\perp_j &:= \min\bigl\{\textrm{wt}(D)\,:\, \textrm{$D$ a right linear $[n,j]$ subcode of $C^\perp$}\}.
\end{align*}
For any subset $X\subseteq E$,
the {\em punctured code} $C\backslash X$ is the right linear code
obtained by deleting the coordinates $X$ from each codeword of~$C$.
Also, $C(X)$ is the right linear subcode of $C$ consisting of all codewords $\mathbf{x}\in C$ with $\textrm{supp}(\mathbf{x}) \subseteq X$.
Note that $k = \dim C = \dim C\backslash X + \dim C(X)$.

Define the function $\rho_C \::\: 2^E \longmapsto \mathbb{N}_0$ by
\[
\rho_C(X) := \dim C\backslash (E-X)\,.
\]
This is the rank function of the vector matroid $M_C=(E,\rho_C)$.
Define $\rho_{C^\perp}$ similarly for $C^\perp$ and note that $\rho^*_C=\rho_{C^\perp}$.
Hence,
\begin{theorem}
\label{thm:code-demi-matroid}
$D_C:=(E,\rho_C,\rho_{C^\perp})$ is a demi-matroid.
\end{theorem}
Consider the numbers $\sigma_i$ and $\tau_j$ for $D_C$.
\begin{proposition}
\label{prop:d_i=sigma_i}
The following identities hold:\vspace{-1mm}
\begin{align*}
d_i &=\overline{\sigma}_i=n-s_{k-i}\,;\\
d^\perp_j&=\overline{\tau}_j=n-t_{n-k-j}\,.
\end{align*}
\end{proposition}
\begin{proof}
Let $D$ be a right linear $[n,i]$ subcode of $C$ with $\textrm{wt}(D)=d_i$,
and set $X=\textrm{Supp}(D)$.
Then
\[
\overline{\rho_C}(X) = k-\rho_C(E-X)\geq
\dim D =i\,,
\]
so $d_i = |X| \geq \overline{\sigma}_i$.

Conversely,
let $X\subseteq E$ be a subset with $\overline{\rho_C}(X)=i$ and $|X|=\overline{\sigma}_i$.
Then $\rho_C(E-X)=k-i$.
Now, $C(X)$ is a right linear subcode of $C$ with
$\dim C(X) = k-\rho_C(E-X) = i$ and
$\textrm{wt}(C(X))\leq |X|$.
Hence, $d_i\leq |X| = \overline{\sigma}_i$.

It follows that $d_i=\overline{\sigma}_i$.
Similarly, $d^\perp_j=\overline{\tau}_j$,
and Lemma~\ref{lem:equ} concludes the proof.
\end{proof}
Set\vspace{-1mm}
\begin{align*}
U_C &:= \{d_1,\ldots,d_k\}\,;\\
V_C &:= \{n+1-d_{n-k}^\perp,\ldots,n+1-d_1^\perp\}\,.
\end{align*}

The result below generalizes Wei's Duality Theorem;
the former specializes to the latter when $R$ is a finite field.
\begin{theorem}
\label{thm:wei}
$U_C\cup V_C=\{1,\ldots,n\}$ and $U_C\cap V_C=\emptyset$.
\end{theorem}
\begin{proof}
The theorem follows
from Theorems~\ref{thm:demi-matroidWei} and~\ref{thm:code-demi-matroid} and Proposition~\ref{prop:d_i=sigma_i}.
\end{proof}

\begin{example}
\label{exa:codeWei3}
{\rm
Consider the linear code $C$ generated by the following binary matrix:
\[
\begin{pmatrix}
1&0&1&0&0\\
0&1&1&0&0\\
0&0&0&1&1
\end{pmatrix}
\]
In this case,
$R=\textrm{GF}(2)$, $E=\{1,2,3,4,5\}$, $n=5$, and~$k=3$.
Furthermore,
$(d_1,d_2,d_3)=(2,3,5)$ and $(d^\perp_1,d^\perp_2)=(2,5)$,
so $U_C=\{2,3,5\}$ and $V_C=\{1,4\}$.
Then $U_C\cap V_C=\emptyset$ and $U_C\cup V_C=\{1,2,3,4,5\}$,
as asserted by Theorem~\ref{thm:wei}.}
\end{example}

\bibliographystyle{tran}

\begin{thebibliography}{10}
\providecommand{\url}[1]{#1}
\csname url@rmstyle\endcsname
\providecommand{\newblock}{\relax}
\providecommand{\bibinfo}[2]{#2}
\providecommand\BIBentrySTDinterwordspacing{\spaceskip=0pt\relax}
\providecommand\BIBentryALTinterwordstretchfactor{4}
\providecommand\BIBentryALTinterwordspacing{\spaceskip=\fontdimen2\font plus
\BIBentryALTinterwordstretchfactor\fontdimen3\font minus
  \fontdimen4\font\relax}
\providecommand\BIBforeignlanguage[2]{{%
\expandafter\ifx\csname l@#1\endcsname\relax
\typeout{** WARNING: IEEEtran.bst: No hyphenation pattern has been}%
\typeout{** loaded for the language `#1'. Using the pattern for}%
\typeout{** the default language instead.}%
\else
\language=\csname l@#1\endcsname
\fi
#2}}

\bibitem{wei91}
V.~K.~Wei,
 ``Generalized Hamming weights for linear codes,''
 \emph{{IEEE} Trans. Inform. Theory}, vol.~37, pp.~1412--1418, 1991.

\bibitem{ChKl03}
W.~Chen and T.~Kl\o ve,
 ``Weight hierarchies of linear codes satisfying the almost chain condition,''
 \emph{Sci. China Ser.~F}, vol.~46, pp.~175--186, 2003.

\bibitem{DoGu01}
S.~T.~Dougherty and T.~A.~Gulliver,
 ``Higher weights and binary self-dual codes,''
 \emph{Electron. Notes Discrete Math.}, vol.~6, 12~pp., 2001.

\bibitem{HeKlMy77}
T.~Helleseth, T.~Kl{\o}ve, and J.~Mykkeltveit,
 ``The weight distribution of irreducible cyclic codes with block lengths $n_1((q^l-1)/N)$,''
 \emph{Discrete Math.}, vol.~18, pp.~179--211, 1977.

\bibitem{HoSh99}
H.~Horimoto and K.~Shiromoto,
 ``A {S}ingleton bound for linear codes over quasi-Frobenius rings,'' in
 \emph{Proc. AAECC-13}, 1999, pp.~51--52.

\bibitem{kloeve92}
T.~Kl{\o}ve,
 ``Support weight distribution of linear codes,''
 \emph{Discrete Math.}, vol.~106/107, pp.~311--316, 1992.

\bibitem{MiCoCo03}
O.~Milenkovic, S.~T.~Coffey, and K.~J.~Compton,
 ``The third support weight enumerators of the doubly-even, self-dual [32,16,8] codes,''
 \emph{{IEEE} Trans. Inf. Theory}, vol.~49, pp.~740--746, Mar. 2003.

\bibitem{milenkovic05}
O.~Milenkovic,
 ``Support weight enumerators and coset weight distributions of isodual codes,''
 \emph{Des.\ Codes Cryptogr.}, vol.~35, pp.~81--109, 2005.

\bibitem{shiromoto03}
K.~Shiromoto,
 ``On $g$-th MDS codes and matroids,''
 \emph{Lect. Notes Comput. Sci.}, vol.~2643, pp.~226--234, 2003.

\bibitem{TsVl95}
M.~A. Tsfasman and S.~G. Vl\u{a}dut, ``Geometric approach to higher weights,''
 \emph{{IEEE} Trans. Inf. Theory}, vol.~41, no.~6, pp.~1564--1588, 1995.

\bibitem{Ashikhimin98}
A.~Ashikhimin,
 ``On generalized Hamming weights for Galois ring linear codes,''
 \emph{Des. Codes Cryptogr.}, vol.~14, pp.~107--126, 1998.

\bibitem{heiseldal08}
B.~Heiseldal,
 \emph{Sammenhenger {M}ellom {K}oder, {M}atroider, {G}rafer og {S}implisielle {K}omplekser},
 Master's thesis, University of Bergen, 2008.

\bibitem{HoSh01}
H.~Horimoto and K.~Shiromoto,
 ``On generalized Hamming weights for codes over finite chain rings,''
 \emph{Lect. Notes Comput. Sci}, vol.~2227, pp.~141--150, 2001.

\bibitem{schaathunthesis}
H.~G.~Schaathun,
 \emph{Support {W}eights in {L}inear {C}odes and {P}rojective {M}ultisets},
 Ph.D. thesis, University of Bergen, 2001.

\bibitem{BrBrShSo07}
D.~Britz, T.~Britz, K.~Shiromoto, and H.~K.~S{\o}rensen,
 ``The higher weight enumerators of the doubly-even, self-dual $[48,24,12]$ code,''
 \emph{{IEEE} Trans. Inf. Theory}, vol.~53, pp.~2567--2571, July 2007.

\bibitem{britz07a}
T.~Britz,
 ``Higher support matroids,''
 \emph{Discrete Math.}, vol.~307, pp.~2300--2308, 2007.

\bibitem{oxley06}
J.~G. Oxley,
 \emph{Matroid {T}heory}, 3rd~ed.\hskip 1em plus 0.5em minus 0.4em\relax
 Oxford: Oxford University Press, 2006.

\bibitem{welsh76}
D.~J.~A. Welsh,
 \emph{Matroid {T}heory}.\hskip 1em plus 0.5em minus 0.4em\relax
 London: Academic Press, 1976.

\bibitem{white86}
 N.~White,
 \emph{Theory of {M}atroids}.\hskip 1em plus 0.5em minus 0.4em\relax
 Cambridge, UK: Cambridge University Press, 1986.

\end{thebibliography}

\end{document}